\documentclass[aps,prx,reprint,superscriptaddress,amsmath,amssymb]{revtex4-2}
\usepackage{graphicx}
\usepackage{amsthm}
\usepackage{thmtools}
\usepackage{url}
\usepackage{tikz}
\usepackage{relsize}
\usepackage{latexsym}
\usepackage{amssymb}
\usepackage{amsmath}
\usepackage{mathtools}
\usepackage{amsmath} 
\usepackage{amsfonts} 
\usepackage{upgreek}
\usepackage{bm}
\usepackage{multirow}
\usepackage{enumitem}
\usepackage{color}
\usepackage[colorlinks, citecolor=blue, linkcolor=blue]{hyperref}

\usepackage{physics}
\usepackage{tcolorbox}
\usepackage{manfnt}
\usepackage{lipsum}
\usepackage[ruled,linesnumbered]{algorithm2e}
\usepackage{tabularx}
\usepackage{nicematrix}
\usepackage{braket}
\usepackage[normalem]{ulem} 
\usepackage{enumitem}

\newcommand{\cC}{\mathcal{C}}

\newcommand{\cO}{ {\cal O} }

\newcommand{\cH}{ {\cal H} }

\newcommand{\ii}{\mathrm{i}}

\newcommand{\U}{\mathrm{U}}

\newcommand{\Var}{\mathrm{Var}}
\newcommand{\Cov}{\mathrm{Cov}}

\newcommand{\appref}[1]{Appendix\,\ref{#1}}
\newcommand{\eqnref}[1]{Eq.\,\eqref{#1}}

\newcommand{\figref}[1]{Fig.\,\ref{#1}}

\newcommand{\rd}{\partial}

\newcommand{\vdagger}{{\vphantom{\dagger}}}

\newcommand{\tcell}[2]{%
  \begin{minipage}[t]{#1}%
  \raggedright
  #2%
  \end{minipage}%
}

\newcommand{\rAngle}{\rangle \hspace{-2pt} \rangle }
\newcommand{\lAngle}{\langle \hspace{-2pt} \langle }
 
\renewcommand{\tr}{\mathrm{tr}}

\theoremstyle{plain} 
\newtheorem{theorem}{Theorem}%[section] 

\newtheoremstyle{axiomwithbreak} % Name of the style
  {3pt}   % Space above
  {3pt}   % Space below
  {} % Body font
  {}      % Indent amount
  {\bfseries} % Theorem head font
  {.}     % Punctuation after theorem head
  {\newline }   % Space after theorem head (newline forces vertical separation)
  {}      % Theorem head spec

\theoremstyle{axiomwithbreak}

\theoremstyle{definition}               
\newtheorem{definition}[theorem]{Definition}

\newtheorem{example}[theorem]{Example} 
\makeatother

\theoremstyle{plain}
\newtheorem{lemma}[theorem]{Lemma}

\newtheorem*{theorem*}{Theorem}

\begin{document}

\title{Charge Scrambling in Strong-to-Weak Spontaneous Symmetry Breaking}

\author{Jong Yeon Lee}
\email{jongyeon@illinois.edu}
\affiliation{Physics Department, University of Illinois at Urbana-Champaign, Urbana, Illinois 61801, USA}
\affiliation{Korea Institute for Advanced Study, Seoul 02455, South Korea}

\date{\today}
\begin{abstract} 
    Strong-to-weak spontaneous symmetry breaking (SWSSB) is diagnosed by nonlinear correlators, but its direct static implication for conserved charge fluctuations is not automatic.
    We show that, for continuous symmetries, long-range R\'enyi-1 correlator, together with a sufficiently rapid approach to its nonzero asymptotic value, forces subsystem charge indefiniteness: the block-charge variance has an extensive lower bound; equivalently, the truncated symmetry expectation has extensive curvature. 
    This gives a precise static fluctuation footprint of charge scrambling.
    We construct examples to show that the implication is conditional and non-reversible: dephased superfluids retain R\'enyi-1 SWSSB with subextensive charge variance when the R\'enyi-1 tail is too slow, while sparse fixed-charge projectors have extensive charge variance but no local charge-transfer Rényi-1 order or long-range conditional mutual information.
    Finally, we introduce a \emph{twist overlap} correlator, which serves as an analogue of charge variance applicable to both discrete and continuous symmetries. 
    This naturally decomposes local block-charge fluctuations into strong- and weak-symmetry channels. We found that the weak-symmetry channel isolates coherent charge fluctuations and is directly related to the Wigner--Yanase skew information. 
    Taken together, these results give a unified understanding for distinguishing nonlinear SWSSB order, local charge indefiniteness, and coherent charge fluctuations.
\end{abstract} 
\maketitle

\section{Introduction}

In mixed quantum states, symmetry comes in two distinct forms. A state may be \emph{weakly symmetric}, meaning that the density matrix is invariant under the symmetry action, or \emph{strongly symmetric}, meaning that the symmetry is respected at the level of each pure-state component and all components carry the same symmetry representation~\cite{Bu_a_2012,PhysRevA.89.022118}. This distinction gives rise to a novel form of spontaneous symmetry breaking in which strong symmetry is lost in the thermodynamic limit while weak symmetry remains intact~\cite{Lee2023PRXQ, LeeU2023,kim2024errorthresholdsykcodes,sala2024spontaneousstrongsymmetrybreaking,lessa2024strongtoweakspontaneoussymmetrybreaking,Ichinose2024,gu2024spontaneoussymmetrybreakingopen, fan2023diagnostics,Lee_2025,Zhang2024PRB,Wightman2025,Weinstein2025,Zhang2025,Feng2025,Ichinose2025,Guo_2025,ziereis2025strongtoweaksymmetrybreakingphases,lu2025holographicdualitybulktopological,vijay2025holographicallyemergentgaugetheory,vijay2025informationcriticalphasesdecoherence,
temkin2025chargeinformedquantumerrorcorrection,
kusuki2026resourcetheoreticquantifiersweakstrong,hauser2026strongtoweaksymmetrybreakingopen}. This phenomenon, dubbed strong-to-weak spontaneous symmetry breaking (SWSSB), was first identified for a discrete $\mathbb{Z}_2$ symmetry~\cite{Lee2023PRXQ} and has since been extended to continuous $U(1)$~\cite{LeeU2023} and fermion parity~\cite{kim2024errorthresholdsykcodes}. Unlike conventional SSB, SWSSB is invisible to ordinary order parameters and instead requires information-theoretic nonlinear diagnostics.

Despite this progress, the physical content of SWSSB remains only partially understood. Several complementary interpretations have emerged. The original work~\cite{Lee2023PRXQ} related SWSSB to the divergent sensitivity of a mixed state to local symmetry-breaking channels. This suggests that, in an SWSSB phase, symmetry information is not stored locally but is delocalized throughout the mixed state: a local symmetry-breaking perturbation can therefore produce a global change. Consistent with this picture, Ref.~\cite{lessa2024strongtoweakspontaneoussymmetrybreaking} showed that the effect of a local symmetry-breaking measurement in an SWSSB phase cannot be reversed by any finite-depth local symmetric channel. For continuous symmetries, SWSSB has also been associated with dynamics: a strongly symmetric Lindbladian with SWSSB steady states was shown to support a long-lived diffusive hydrodynamic mode~\cite{LeeU2023,Moudgalya2024}.

\begin{figure}[!t]
    \centering
    \includegraphics[width=0.99\linewidth]{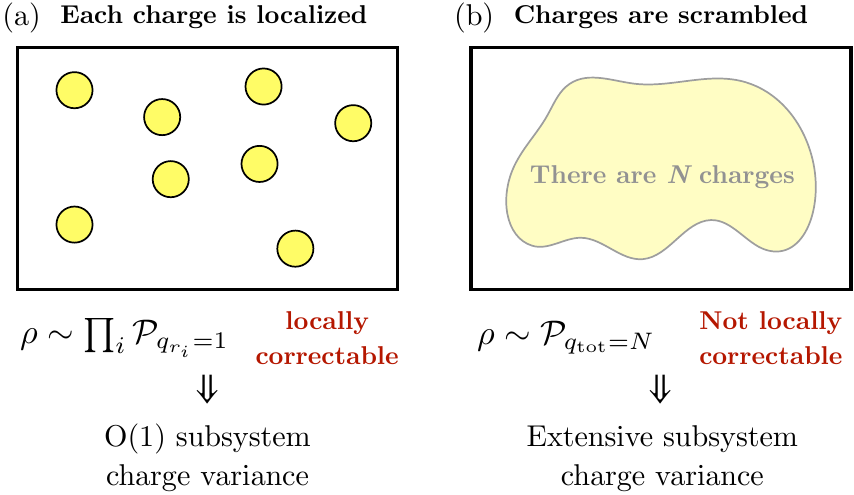}
    \caption{{\bf Unscrambled and scrambled states}. (a)  When the charge information is stored locally, the state is strongly symmetric and a local perturbation can be repaired locally. (b) In the uniform fixed-charge example, the charge is delocalized and the state exhibits SWSSB: a local symmetry-breaking channel produces an effect that cannot be reversed by a finite-depth local symmetric channel.}
    \label{fig:scramble}
    \vspace{-10pt}
\end{figure}

 \begin{table*}[!t]
  \footnotesize
  \renewcommand{\arraystretch}{1.25}
  \begin{tabular*}{\textwidth}{@{\extracolsep{\fill}}lllll@{}}
    \hline\hline
    \tcell{0.16\textwidth}{State} &
    \tcell{0.20\textwidth}{Long-range diagnostic} &
    \tcell{0.095\textwidth}{$\mathrm{Var}(Q_A)$} &
    \tcell{0.13\textwidth}{Charge coherence} &
    \tcell{0.32\textwidth}{} \\
    \hline

    \tcell{0.16\textwidth}{Uniform ensemble}
    &
    \tcell{0.20\textwidth}{$C(r)=0$ \\ $R(r)= \nu(1-\nu)$}
    &
    \tcell{0.095\textwidth}{extensive}
    &
    \tcell{0.13\textwidth}{$I_{\rm WY}=0$}
    &
    \tcell{0.32\textwidth}{SWSSB state with scrambled charge yet no ordinary LRO.}
    \\

    \tcell{0.16\textwidth}{Dicke state}
    &
    \tcell{0.20\textwidth}{$C(r) = \nu(1-\nu)$ \quad $R(r) = \nu^2(1-\nu)^2$}
    &
    \tcell{0.095\textwidth}{extensive}
    &
    \tcell{0.13\textwidth}{$I_{\rm WY}=\mathrm{Var}(Q_A)$}
    &
    \tcell{0.32\textwidth}{Extensive charge variance is compatible with genuine symmetry breaking.}
    \\

    \tcell{0.16\textwidth}{Superfluid or XY ferromagnet}
    &
    \tcell{0.20\textwidth}{Long-range $C(r)$ and $R(r)$ with $1/r^{d-1}$ tail.}
    &
    \tcell{0.095\textwidth}{subextensive}
    &
    \tcell{0.13\textwidth}{$I_{\rm WY}=\mathrm{Var}(Q_A)$}
    &
    \tcell{0.32\textwidth}{Ordinary SSB example motivating the dephased construction below.}
    \\

    \tcell{0.16\textwidth}{Dephased superfluid}
    &
    \tcell{0.20\textwidth}{$C(r)=0$, while $R(r) \sim R_\infty + a/r^{d-1}$.}
    &
    \tcell{0.095\textwidth}{subextensive}
    &
    \tcell{0.13\textwidth}{$I_{\rm WY}=0$}
    &
    \tcell{0.32\textwidth}{SWSSB counterexample showing that fast R\'enyi convergence is needed.}
    \\

    \tcell{0.16\textwidth}{Sparse projector $\rho_\Omega$}
    &
    \tcell{0.20\textwidth}{no long-range $R(r)$; no SWSSB.}
    &
    \tcell{0.095\textwidth}{extensive}
    &
    \tcell{0.13\textwidth}{$I_{\rm WY}=0$}
    &
    \tcell{0.32\textwidth}{Counterexample showing that the converse of Theorem~\ref{thm:main} is false.}
    \\[9pt]

    \hline\hline
  \end{tabular*}
  \caption{\label{tab:comparison_examples}
  Compact comparison of the representative states. $C(r)$ and $R(r)$ denote the ordinary and R\'enyi-1 correlators. The behavior of $\textnormal{Var}(Q_A)$ is consistent with Theorem~\ref{thm:main}, which links sufficiently fast convergence of long-range $R(r)$ to extensive variance. For the sparse projector state, the conditional mutual information diagnostic also vanishes, see \appref{app:CMI}.}
\end{table*}

This dynamical connection raises a sharper question. The hydrodynamic mode is a statement about linearized charge dynamics, whereas SWSSB is diagnosed by nonlinear correlators of the mixed state. Thus one should not expect SWSSB alone to imply a hydrodynamic response without identifying the additional static fluctuation content carried by the steady state. At the microscopic level, the long-lived modes found in strongly symmetric Lindbladians are closely tied to the low-momentum static charge structure factor of the steady state~\cite{LeeU2023}. At the phenomenological field-theory level, related statements are often phrased in terms of a charge susceptibility~\cite{Akyuz_2024,luca2025,Huang2025}. These two quantities, however, are conceptually distinct. The static structure factor is an equal-time property of the state alone, while the susceptibility is a causal response function and therefore depends on additional dynamical or Hamiltonian data. In thermal equilibrium they are related by fluctuation-dissipation relations, but no such relation is guaranteed for generic mixed states or nonequilibrium steady states. 
It is therefore important to better understand the purely static part that plays the key role in connecting SWSSB and hydrodynamics.

A motivating example is illustrated in Fig.~\ref{fig:scramble}. Consider first a product state with a definite local charge configuration $\otimes_i P_{q_i}^{r_i}$, where $P_{q_i}^{r_i}$ projects onto charge $q_i$ at site $r_i$. In this case, a local perturbation can in principle be diagnosed and repaired locally, and subsystem charge fluctuations remain subextensive. By contrast, the fixed-charge state $\rho_Q\propto P_Q$, where $P_Q$ projects onto the total-charge-$Q$ sector, delocalizes charge information over all microscopic configurations compatible with the global constraint. The charge information is globally encoded with nontrivial conditional mutual information~\cite{lessa2024strongtoweakspontaneoussymmetrybreaking, yang2025topologicalmixedstatesphases}, and exhibits extensive subsystem charge fluctuations. The state gives a prototypical example of $U(1)$ SWSSB~\cite{LeeU2023}, suggesting a close relation between SWSSB and local charge indefiniteness. Related studies have also pointed to such a connection from a charge-sharpening perspective~\cite{PhysRevLett.129.200602, PhysRevX.12.041002, Singh_2026, vijay2025holographicallyemergentgaugetheory}. However, a general equivalence remained elusive.

In this work we isolate a static fluctuation component of charge scrambling. For a continuous symmetry, the charge in a region $A$ is locally indefinite when the truncated symmetry expectation $\tr \rho e^{i \theta Q_A}$ decays with the size of $A$; infinitesimally this is measured by the variance of $Q_A$. 
We prove that a long-range R\'enyi-1 correlator, together with sufficiently rapid convergence to its asymptotic value, enforces such subsystem charge indefiniteness. The converse, however, is not true in general. We construct examples with extensive charge fluctuations but no SWSSB, as well as examples in which SWSSB does not enforce an extensive charge variance because the approach to the long-distance limit is slow. 
These results clarify the precise conditions under which SWSSB implies a nonzero static structure factor.

We then extend the discussion to discrete symmetries by introducing a R\'enyi-1 \emph{twist overlap}, a discrete analogue of charge variance. Much like the conventional order-disorder competition, the twist overlap competes with the nonlinear correlator and thereby provides a complementary diagnostic of SWSSB. 
It also connects directly to the Wigner-Yanase skew information, which isolates the coherent part of block-charge fluctuations~\cite{WY1963,Luo2005,Hansen2008}. Together, these results provide a unified framework in which the R\'enyi-1 correlator diagnoses nonlinear symmetry breaking, subsystem charge variance measures the total fluctuation budget, and the Wigner-Yanase skew information extracts its genuinely coherent component.

% \vspace{5pt} \emph{Main result}. 

\section{Main Result}

Given a density matrix $\rho$, its canonical purification can be written as
\begin{align}
  \Vert \sqrt{\rho} \rAngle := (\sqrt{\rho} \otimes \mathbf{1}) \sum_i | i \rangle_L | \bar{i} \rangle_R,
\end{align}
which lives in the doubled Hilbert space. 
In the doubled Hilbert space, the strong symmetry $G$ becomes doubled, i.e., $G_L \times G_R$.
When $G$ is $\U(1)$ symmetry, we can redefine two generators as the sum and difference of the left and right generators (See \appref{app:doubled} for conventions): 
\begin{align}
  Q^+ := Q_{L} + Q_{R}, \quad Q^- := Q_{L} - Q_{R}.
\end{align}
Therefore $Q^-$ is the weak symmetry generator of $\rho$~\footnote{But it does not mean that $Q^+$ is the strong symmetry generator.}:
\begin{align}
    e^{i \theta Q^-} \Vert \sqrt{\rho} \rAngle = (U^{\vphantom{*}} \otimes U^*) \Vert \sqrt{\rho} \rAngle 
    = \Vert \sqrt{U \rho U^\dagger} \rAngle.
\end{align} 
We consider the charge $Q_A$ contained in a subsystem $A$. Consider a truncated symmetry action $U_A = e^{i \theta Q_A}$, where $Q_A := \sum_{r \in A} q_r$ measures the total charge within $A$. In general, how \emph{scrambled} the charge is can be diagnosed by how well-defined the truncated symmetry action $U_A$ is, whose expectation value is given by
\begin{align} \label{eq:truncated_action}
  |\langle e^{i \theta Q_A} \rangle| = 
  e^{-\frac{\theta^2}{2} \textnormal{Var}(Q_A) + \cO(\theta^4)},
\end{align} 
where $\textnormal{Var}(Q_A) := \langle Q_A^2 \rangle - \langle Q_A \rangle^2$. Therefore, the subsystem charge variance $\textnormal{Var}(Q_A)$ quantifies the local indefiniteness of the charge in $A$. In what follows, we establish the relation between the behavior of $\textnormal{Var}(Q_A)$ and the nonlinear diagnostics of SWSSB.

Consider a charge \emph{transfer} operator $T_{xy}$ supported near points $x$ and $y$, with $x \in A$ and $y \notin A$ such that 
\begin{align}
  [Q_A, T_{xy}] = Q_A T_{xy} - T_{xy} Q_A = T_{xy}.
\end{align}
In other words, $T_{xy} \sim b_x^\dagger b_y^\vdagger$ transfers a charge from $y$ to $x$. Then we define the R\'enyi-1 correlator as: 
\begin{align}
  R_{xy} := \tr \big( \sqrt{\rho} T^\vdagger_{xy} \sqrt{\rho} T_{xy}^\dagger ) = \lAngle \sqrt{\rho} \Vert T^{\vphantom{*}}_{xy} \otimes T_{xy}^* \Vert \sqrt{\rho} \rAngle.
\end{align}
The long-range order in $R_{xy}$, which indicates the $\U(1)^+$ SSB in the purified state, implies the SWSSB of the original density matrix~\cite{Weinstein2025}. The local $\U(1)^+$ order parameter $\tr(\sqrt{\rho} b_x^\dagger \sqrt{\rho} b_x)\,{=}\,0$~\footnote{$T$ denotes the transpose in the occupation-number basis used to define the vectorization. In this basis $c_x^T=c_x^\dagger$ because the matrix elements of $c_x$ are real. Thus $c_x^T$ creates charge on the second copy: $[Q^T,c_x^T]=c_x^T$. The statement is basis-dependent; what is basis-independent is the doubled-space construction once the transpose and the second-copy charge operator are transformed consistently.} due to strong symmetry.
% $\lAngle \sqrt{\rho} \Vert b_x^\dagger \otimes b_x^T \Vert \sqrt{\rho} \rAngle = 0$

The fluctuations of the purified and original states are related by the following lemma (proof in \appref{app:proofs}):
\begin{lemma} \label{lem:decomposition_var}
  The charge variance in a region $A$ can be decomposed as the sum of charge variances of the purified state in the doubled Hilbert space: 
  \begin{align}
    \textnormal{Var}(Q_A) = \frac{1}{4} \textnormal{Var}(Q_A^+) + \frac{1}{4} \textnormal{Var}(Q_A^-). 
  \end{align}
\end{lemma}
Thus, it suffices to establish extensive fluctuations in one of the doubled charges. We will show that long-range order in $R_{xy}$, together with sufficiently rapid convergence toward its asymptotic value, implies $\operatorname{Var}(Q_A^+)\,{\sim}\,|A|$. By the decomposition above, this in turn implies extensive $\operatorname{Var}(Q_A)$, which is the signature of scrambled charge.

\begin{theorem}[SWSSB and lower bound on charge variance] \label{thm:main}
  Consider a strongly $\U(1)$-symmetric state $\rho$ with a long-range R\'enyi-1 correlator of the charge transfer operator. Then there exists $c>0$ such that 
  \begin{align}
    \textnormal{Var}(Q_A) \geq \frac{c |A|^2}{\textnormal{Var}_\varphi(Y_A)},
  \end{align}
  where $\textnormal{Var}_\varphi(Y_A)$ is the variance of the SWSSB order parameter in the extremal SSB state.
  In particular, if the R\'enyi-1 correlator converges to its asymptotic value faster than $1/r^d$, where $d$ is the spatial dimension, then the subsystem charge variance scales extensively, i.e., $\textnormal{Var}_\varphi(Y_A) \leq c' |A|$ for some $c'$ such that
  \begin{align}
    \textnormal{Var}(Q_A) \geq (c/c') |A|.
  \end{align}
  This also implies a finite nonzero small-momentum limit of the charge structure factor $S(k \rightarrow 0^+)$. 
  \begin{proof}
    The long-range R\'enyi-1 correlator indicates that $\Vert \sqrt{\rho} \rAngle$ has SSB of $U(1)$ symmetry generated by $Q^+$. The order parameter for this SSB in the purified state can be written as $O_r = b_{r,L}^\dagger b_{r,R}^T$~\footnote{Generalization to any $\cO(1)$ charge transfer case is straightforward.}, which is a local operator carrying charge-2 under $\U(1)^+$ and neutral under $\U(1)^-$. Define the summed order parameter as 
    \begin{align}
      \Phi_A := \sum_{r \in A}  O_r. 
    \end{align}
    To proceed, we define 
    \begin{align}
      X_A := \frac{\Phi_A + \Phi_A^\dagger}{2}, \quad Y_A := \frac{\Phi_A - \Phi_A^\dagger}{2i}.
    \end{align}
    Since $O_r$ carries charge-2 under $Q_A^+$ and neutral under $Q_A^-$, we obtain
    \begin{align}
      [Q_A^+, Y_A ] = -2 i X_A.
    \end{align}
    The Robertson inequality~\cite{Robertson1929} says that for any two operators $A$ and $B$, the expectation value of the commutator $[A,B]$ is bounded by the square-root of the product of the variances of $A$ and $B$. Therefore,
    \begin{align} \label{eq:inequality_proof}
      | \lAngle X_A \rAngle |^2 = \frac{1}{4} | \lAngle [Q_A^+, Y_A ] \rAngle |^2 \leq \textrm{Var}(Q_A^+)  \textrm{Var}(Y_A).
    \end{align}
    In other words, 
    \begin{align} \label{eq:inequality_intermediate}
      \textrm{Var}(Q_A^+) \geq \frac{|\lAngle X_A \rAngle|^2}{\textrm{Var}(Y_A)}.
    \end{align}
    However, this form is not sufficient to get the desired lower bound since our canonically purified state has $\lAngle X_A \rAngle = 0$ and $\textnormal{Var}(Y_A) \sim |A|^2$.

    To obtain a meaningful bound, we want to apply the inequality in \eqnref{eq:inequality_proof} to an \emph{extremal state}. Note that $\Vert \sqrt{\rho} \rAngle$ has $U(1)^+$ SSB while having vanishing local order parameter for $(X_A,Y_A)$. Therefore, the state should be given as the uniform superposition of extremal SSB states $\{ \Vert m, \varphi \rAngle\}$ in which 
    \begin{align}
      \lim_{A \rightarrow \infty} \frac{1}{|A|}  \lAngle m, \varphi \Vert \Phi_A \Vert m, \varphi \rAngle = m e^{i \varphi}.
    \end{align}
    Accordingly, its thermodynamic-limit local expectation values admit the standard phase-average decomposition over extremal SSB states:
    \begin{align}
      \lAngle \cO \rAngle_{\Vert \sqrt{\rho} \rAngle}
      =
      \int \frac{d\varphi}{2\pi}\,
      \lAngle \cO \rAngle_{\Vert m,\varphi \rAngle},
    \end{align}
    for any quasilocal observable $\cO$. Equivalently, the finite-volume symmetric state may be viewed as a uniform superposition of phase-selected states, with cross terms vanishing in the thermodynamic limit because distinct phases become orthogonal. Applying this decomposition to $(Q_A^+)^2$ and $Q_A^+$ for a fixed finite region $A$ gives
    \begin{align}
      \textnormal{Var}(Q_A^+)_{\Vert \sqrt{\rho} \rAngle} = \int \frac{d\varphi}{2\pi} \, \textnormal{Var}_\varphi(Q_A^+),
    \end{align}
    since the expectation value of $Q_A^+$ is the same in all extremal states and in the canonically purified state, so there is no additional phase-mixing contribution to the variance.
    % different extremal SSB states are orthogonal in the thermodynamic limit. 
    Here $\textnormal{Var}_\varphi(Q_A^+)$ is the variance of $Q_A^+$ in the extremal SSB state $\Vert m, \varphi \rAngle$.
    Due to the symmetry, the variance is independent of $\varphi$, so we can identify that 
    \begin{align}
      \textnormal{Var}(Q_A^+)_{\Vert \sqrt{\rho} \rAngle} = \textnormal{Var}_\varphi(Q_A^+).
    \end{align}
    In the extremal SSB state, we have     
    \begin{align}
      |\lAngle X_A \rAngle|^2 = c |A|^2,
    \end{align}
    for some $c\,{>}\,0$. Finally, we use \eqnref{eq:inequality_intermediate} for the phase-selected extremal state and, using Lemma~\ref{lem:decomposition_var}, we obtain the desired lower bound:
    \begin{align}
      \textrm{Var}(Q_A) \geq \frac{1}{4}\textrm{Var}(Q_A^+) \geq \frac{c |A|^2}{\textrm{Var}_\varphi(Y_A)},
    \end{align}
    for some system-size-independent constant $c > 0$. 
    
    For the second part that uses rapid convergence of R\'enyi-1 correlator, consider $Y_r := (O^\vdagger_r - O_r^\dagger)/(2i)$. For simplicity, assume translation invariance, although this is not strictly necessary. The variance of $Y_A$ can be written as the summation of connected correlation functions:
\begin{align}
  \textrm{Var}(Y_A) = |A| \lAngle Y_r^2 \rAngle_c + \hspace{-3pt}\sum_{r \neq r' \in A} \hspace{-3pt} \lAngle Y_r Y_{r'} \rAngle_c.
\end{align}
The first term scales as $|A|$. If the connected correlation function for $Y_r$ decays faster than $1/r^d$, where $d$ is the spatial dimension, the second term also scales as $|A|$:
\begin{align}
  \sum_{r \neq r' \in A} \hspace{-3pt} \lAngle Y_r Y_{r'} \rAngle_c \leq |A| \sum_{r \neq 0} \frac{1}{r^{d+\epsilon}} \leq c' |A|.
\end{align}
Note that the tail behavior of the R\'enyi-1 correlator must match the connected correlators in the extermal state in the SWSSB phase.
Therefore, the variance of $Y_A$ scales as $|A|$, and the subsystem charge variance scales extensively. 

The static structure factor is defined as 
\begin{align}
  S(k) := \frac{1}{V} \langle Q_k Q_{-k} \rangle_\rho,
\end{align}
where we are interested in the $k \rightarrow 0^+$ limit. While $S(0)\,{=}\,0$ for a fixed-charge ensemble, $\lim_{k \rightarrow 0^+} S(k)$ is finite if the subsystem charge variance scales extensively; see \appref{app:structure_factor}.
  \end{proof}
\end{theorem}

\begin{example}[Uniform fixed-charge ensemble]
   Consider the incoherent mixture of all possible charge configurations given a total charge constraint~\cite{LeeU2023}. In fact, for a generic strongly U(1)-symmetric local Lindbladians, this corresponds to the unique steady state~\cite{Yoshida_2024}. For simplicity, we consider hard-core bosons. Let $\nu = Q/V$ be the charge density. Let $d_Q = \binom{V}{Q}$ be the number of configurations with total charge $Q$. Then
   \begin{align}
    \sigma = \frac{1}{d_Q} P^Q = \frac{1}{d_Q} \sum_{ \sum_r \hspace{-2pt} q_r = Q}  \bigotimes_{r} P_r^{q_r}.
   \end{align}
   The ordinary and R\'enyi-1 correlators are given by
   \begin{align}
      C(r) = 0, \quad R(r) = \nu(1-\nu) + \cO(V^{-1}),
   \end{align}
   where $C(r) = \tr \rho T_{xy}$ and $R(r) = \tr \sqrt{\rho} T_{xy} \sqrt{\rho} T_{xy}^\dagger$. Therefore, the state has SWSSB. Furthermore, $R(r)$ converges immediately to its asymptotic value $\nu(1-\nu)$. Therefore, Theorem~\ref{thm:main} applies, and the charge variance should scale extensively.  Indeed, in the thermodynamic limit with $|A|/V\rightarrow 0$, the subsystem charge variance is 
   \begin{align}
    \textrm{Var}(Q_A) = \nu(1-\nu) |A|.
   \end{align}
\end{example}

\begin{example}[Dicke state]
  Consider a pure-state wavefunction given as the superposition of all possible charge configurations given a total charge sector:
  \begin{align}
    |\psi \rangle =  \frac{1}{\sqrt{d_Q}} 
    \sum_{ \sum_r \hspace{-2pt} q_r = Q}  \bigotimes_{r} |q_r \rangle
  \end{align}
  In the thermodynamic limit $V \rightarrow \infty$, for $r \neq 0$,
  \begin{align}
    C(r) = \nu(1-\nu), \quad  R(r) = \nu^2(1-\nu)^2
  \end{align}
  up to $\cO(V^{-1})$ corrections.
  With long-range order in the ordinary correlator, the state has genuine SSB~\footnote{More precisely, this is the  superposition of extremal states as the local order parameter vanishes.}.
  The absence of $r$ dependence implies both $C(r)$ and $R(r)$ rapidly converge to asymptotic values. Therefore, Theorem~\ref{thm:main} applies. Again in the thermodynamic limit, $\textrm{Var}(Q_A) = \nu(1-\nu) |A|$. 
\end{example}

\begin{example}[Superfluid state]
  Fast convergence of R\'enyi-1 correlator is crucial to show the extensive scaling of charge variance. 
  Consider a fixed-charge pure state with $\U(1)$ SSB, i.e., a superfluid or XY ferromagnet. The state has both long-range ordinary and R\'enyi-1 correlators, where both $C(r)$ and $R(r)$ approach their asymptotic values with $1/r^{d-1}$ tail in $d$ spatial dimensions due to gapless Goldstone modes. With the slow tail, Theorem~\ref{thm:main} does not apply.
  Indeed, the state has a subextensive charge variance~\cite{Fluctuations2012}; its static charge structure factor $S(k) \sim |k|$, which vanishes in the $k \rightarrow 0$ limit.
\end{example}

\begin{example}[Dephased superfluid state]
	If we fully dephase the aforementioned fixed-charge superfluid state in the occupation basis, we obtain a strongly symmetric state with the same diagonal elements as the original superfluid state. As this state is diagonal, the off-diagonal correlator $C(r)$ vanishes, and there is no ordinary SSB. On the other hand, for hard-core bosons whose ground-state amplitudes can be chosen nonnegative in the occupation basis, the R\'enyi-1 correlator of the dephased state equals the ordinary one-body correlator of the original superfluid wavefunction. Indeed, writing $|\psi\rangle=\sum_\eta \psi_\eta |\eta\rangle$ with $\psi_\eta\geq 0$ and $\rho_{\rm deph}=\sum_\eta |\psi_\eta|^2|\eta\rangle\langle \eta|$, one obtains
	  \begin{align}
	    R^{\rm deph}_{xy}
	    = \hspace{-10pt} \sum_{ (\eta_x,\eta_y) = (0,1) } \hspace{-10pt} \psi_{\eta^{xy}}\psi_\eta
	    = \langle \psi | b_x^\dagger b_y |\psi\rangle,
	  \end{align}
	  where $\eta^{xy}$ is obtained from $\eta$ by moving one particle from $y$ to $x$. Thus the resulting state still has SWSSB; furthermore, the R\'enyi-1 correlator converges slowly to its asymptotic value, inheriting the same power-law tail as the original superfluid state. Accordingly, Theorem~\ref{thm:main} does not apply. 
  Indeed, the subsystem charge variance is still subextensive, as the dephasing does not change the diagonal elements of the density matrix.  
\end{example}

Finally, we make two important points. First, even if the R\'enyi-1 correlator has a tail decaying slower than $1/r^{d}$, the subsystem charge variance can still scale extensively. Such a state can be constructed by taking a convex combination of the Dicke state and the superfluid state. Second, the converse of Theorem~\ref{thm:main} does not hold, as shown by the following example. 
\begin{example}[Sparse projector state] \label{ex:converse_not_true}
  The main counterexample to the converse is provided by a typical sparse projector inside a fixed-charge sector. Consider hard-core bosons on $V$ sites at fixed total charge $Q=\nu V$, and choose a random subset $\Omega$ of the occupation basis states in the fixed-charge sector, with $V^2 \ll r(V)=|\Omega|$ and $\log r(V)=o(V)$. Then define the mixed state
  \begin{align}
  \rho_\Omega=\frac{1}{r(V)}\sum_{\mathbf n\in\Omega} |\mathbf n\rangle \langle \mathbf n|.
  \end{align}
  As proven in \appref{app:countexample}, this state almost surely has extensive charge variance in the thermodynamic limit. The absence of R\'enyi-1 order comes from a different, local-connectivity criterion. For a diagonal projector, the R\'enyi-1 correlator of a hopping operator $b_i^\dagger b_j^\vdagger$ counts pairs of configurations in the support that are related by this local hopping move. Thus a configuration $\mathbf n\in\Omega$ contributes only if the hopped configuration $T_{j\to i}\mathbf n$ also lies in $\Omega$. Since $\Omega$ contains only subexponentially many configurations inside the exponentially large fixed-charge sector of size $d_Q$, this locally moved partner is almost never present. Quantitatively,
  \begin{align}
  \mathbb E_\Omega R^{(\Omega)}_{ij}
  &\sim \nu(1-\nu) \frac{r(V)}{d_Q}.
  \end{align}
  Equivalently, the expected number of surviving hopping edges is $r(V)\mathbb E_\Omega R^{(\Omega)}_{ij}\sim \nu(1-\nu)r(V)^2/d_Q$, which vanishes as $V \rightarrow \infty$ since $d_Q$ is exponentially in $V$. 
  Since this edge count is a nonnegative integer, it follows that, with probability approaching one, for any fixed pair $i\neq j$,
  \begin{align}
  R_{ij}^{(\Omega)} = 0.
  \end{align}
  More generally, any fixed finite-support $\cO(1)$ charge-transfer operator has a vanishing R\'enyi-1 correlator. The block-charge variance survives because it probes only the distribution of $Q_A$ among the sampled configurations, not the local connectivity of the support under charge-transfer moves. Furthermore, the state has vanishing conditional mutual information for a large enough shielding region; see \appref{app:CMI}. Therefore $\rho_\Omega$ has no SWSSB, despite having extensive subsystem charge variance.
\end{example}
\noindent The representative continuous-symmetry examples discussed above are summarized in Table~\ref{tab:comparison_examples}.

% \vspace{5pt} \noindent {\bf Remark 2.} For discrete global symmetries, localized charge is not well-defined. Instead, for a given subsystem one may define the following quantity. For a discrete Abelian symmetry group $G$, assume $U_g := \prod_i U_{g,i}$. Then one can define a truncated symmetry operator $U_{g,A}:= \prod_{i \in A} U_{g,i}$, which is also called a disorder operator. In this case, one is interested in the R\'enyi-1 correaltor of the disorder operator:
% \begin{align}
%     D_{A} := \tr \big( \sqrt{\rho} U_{g,A}^\vdagger \sqrt{\rho} U_{g,A}^\dagger \big).
% \end{align}
% For $\U(1)$ symmetry, note that $U_{\theta,i} = e^{i q_i \theta}$; accordingly, for small enough $\theta$, we observe that 
% \begin{align}
%     D_{A}^{\rm U(1)} \approx 1 + \theta^2 \tr \big( \sqrt{\rho} Q_A \sqrt{\rho} Q_A \big)
% \end{align}

% \vspace{5pt} \emph{Twist overlap and skew information}.

\begin{figure}[!t]
    \centering
    \includegraphics[width=0.73\linewidth]{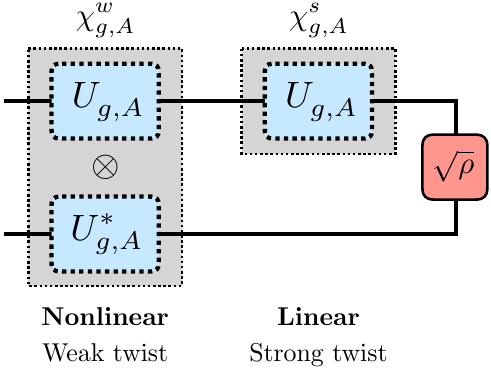}
    \caption{{\bf Strong and weak disorders}. In the canonically purified state, there are disorder operators with respect to strong and weak symmetries. The strong one reduces down to linear observable, while the weak one stays nonlinear quantity.}
    \label{fig:disorder}
    \vspace{-10pt}
\end{figure}

\section{Twist overlap}

In the previous section, we identified conditions under which the subsystem charge variance is extensive, providing a natural diagnostic of charge scrambling. This diagnostic, however, has two issues. First, it does not distinguish whether the fluctuation has purely classical or quantum coherent origin. Second, it is intrinsically tied to continuous symmetries. The underlying idea nevertheless extends to discrete symmetries once one replaces the subsystem charge by the corresponding truncated symmetry action. Indeed, the subsystem charge variance may be viewed as probing the response to the truncated symmetry transformation in \eqnref{eq:truncated_action}. For a group element $g \in G$, let
\begin{align}
  U_{g,A} := \prod_{i \in A} U_{g,i}
\end{align}
denote the symmetry action restricted to a subregion $A$. This operator is often referred to as a disorder operator~\cite{KadanoffCeva1971, Fradkin2016}. We now introduce its R\'enyi-1 analogue, which is visualized in \figref{fig:disorder}.
\begin{definition}[Twist overlap] \label{def:twist_overlap}
  Given a state $\rho$, we define a \emph{twist overlap} $\chi_{g,A}^{s/w}$ on subregion $A$ as
\begin{align}
  \chi^w_{g,A} &:= \lAngle \sqrt{\rho} \Vert U^{\vphantom{*}}_{g,A} \otimes U_{g,A}^* \Vert \sqrt{\rho} \rAngle = \tr\big( \sqrt{\rho} U_{g,A}^\vdagger \sqrt{\rho} U_{g,A}^\dagger  \big), \nonumber \\
  \chi^s_{g,A} &:= \lAngle \sqrt{\rho} \Vert U^{\vphantom{*}}_{g,A} \otimes \mathrm{Id} \Vert \sqrt{\rho} \rAngle= \tr\big( \rho U_{g,A} \big).
\end{align}
Both correspond to the disorder operators in the purified state: $\chi^s$ against strong symmetry, and $\chi^w$ against weak symmetry. 
However, there is an important subtlety. 
Following the discussion in Example.~\ref{ex:converse_not_true}, the volume-law decay of the strong twist overlap $\chi^s_{g,A}$ is not sufficient to diagnose SWSSB. This is contrasted with one-dimensional ground states of local Hamiltonian, where volume-law decay of the disorder operator indicates ordering~\cite{Levin_2020}. On the other hand, the weak twist overlap $\chi^w_{g,A}$ is a genuinely nonlinear diagnostic. In particular, it captures the local symmetry charge indefiniteness originated from coherent part of the fluctuation.
\end{definition}

\begin{example}[Subtlety of the disorder operator]
  Consider the prototypical model studied in Ref.~\cite{Lee2023PRXQ}, namely a two-dimensional spin system under strongly $\mathbb{Z}_2$-symmetric decoherence. 
  Let $U_{A} = \prod_{i \in A} X_i$. 
  Under the Kramers-Wannier duality, $\chi^s_A$ maps to the loop expectation value along the boundary of $A$ in the decohered toric code model~\cite{fan2023diagnostics}. 
  However, upon increasing the decoherence strength, which maps to the onsite $Z$ error, the loop expectation value always decays exponentially with the perimeter of $A$, even in the SWSSB phase. Thus, there is no universal order-vs-disorder dichotomy~\cite{Levin_2020}. On the other hand, $\chi^w$ remains insensitive to the decoherence.
\end{example}

The weak-symmetry twist overlap $\chi^w$ is closely tied to the Wigner-Yanase skew information~\cite{WY1963,Luo2005,Hansen2008}. Given an additive conserved quantity $O$, the Wigner-Yanase skew information is defined as
\begin{align} \label{eq:YW}
  I_{\rm WY}(\rho, O) &: 
  = - \frac{1}{2} \textnormal{tr}\Big( [O, \sqrt{\rho} ]^2 \Big).
\end{align}
It measures the noncommutativity between $\rho$ and $O$, and hence isolates the coherent, or asymmetric, part of the $O$-fluctuation. In particular, it ignores classical spread among different $O$-sectors: $I_{\rm WY}(\rho,O)=0$ whenever $[\rho,O]=0$, while for pure states it reduces to the ordinary variance, $I_{\rm WY}(\rho,O)=\textnormal{Var}_{\rho}(O)$. Under symmetry-respecting ($O$-covariant) channels, it monotonically decreases, acting like a measure of quantum coherence resource. Indeed, it constrains the amount of coherence relative to $O$ that can be distilled~\cite{Marvian_2013}. Finally, the skew information has metrological implications by providing two-sided bounds on the quantum Fisher information~\cite{Gibilisco2009}.

% The twist overlap of the weak symmetry $\chi^w$ is closely tied to the Wigner-Yanase skew information~\cite{WY1963,Luo2005,Hansen2008}. Given an additive conserved quantity denoted as $O$, it is defined as
% \begin{align} \label{eq:YW}
%   I_{\rm WY}(\rho, O) &: 
%   = - \frac{1}{2} \textnormal{tr}\Big( [O, \sqrt{\rho} ]^2 \Big).
% \end{align}
% As a measure of the noncommutativity between $\rho$ and $O$, $I_{\rm WY}(\rho, O)$ ignores classical spread and detects coherence/asymmetry relative to the generator. Being monotone under channels, it has been proposed to quantify the quantum coherence in the resource theory of asymmetry~\cite{Marvian_2013}, providing bound on how much of the coherence with respect to $O$ is distillable; furthermore, the skew information and quantum Fisher information have mutually bounding relation with metrological implication~\cite{Luo2005, Gibilisco2009}.

For $U(1)$ symmetry, where $g$ is parametrized by an angular variable $\theta$, differentiating $\chi^w_{\theta,A}$ gives
\begin{align}
  \rd_\theta \chi^w_A &= i \tr \Big( \sqrt{\rho} e^{i \theta Q_A} \big[Q_A, \sqrt{\rho} \big] e^{-i \theta Q_A}  \Big) \mapsto_{\theta = 0} 0, \nonumber \\
  \rd_\theta^2 \chi^w_A  &= - \tr \Big( \sqrt{\rho} e^{i \theta Q_A} \big[ Q_A, \big[Q_A, \sqrt{\rho} \big] \big] e^{-i \theta Q_A}  \Big) \nonumber \\
  & \mapsto_{\theta = 0} \tr \Big( [Q_A, \sqrt{\rho}]^2 \Big) = -2 I_{\rm WY}(\rho, Q_A).
\end{align}
Applying the form derived in \eqnref{eq:truncated_action} for $Q_A^-$ in the doubled Hilbert space, together with the fact that $\lAngle Q_A^- \rAngle = 0$ for weakly symmetric states,   
% since $\lAngle Q_A^L \rAngle = \lAngle Q_A^R \rAngle = \tr \rho Q_A$, 
we can rewrite the Wigner-Yanase skew information as
\begin{align}
  I_{\rm WY}(\rho, Q_A) = \frac{1}{2} \textnormal{Var}(Q_A^-).
\end{align} 
Therefore, the fluctuation of $Q_A^-$ measures the \emph{coherent} part of the charge variance; somewhat intuitively, the fluctuation of $Q_A^-$ cannot be larger than the fluctuation of $Q_A^+$ as proven in \appref{app:proofs}:
\begin{lemma}[$Q^+$ fluctuations bound $Q^-$ fluctuations]\label{lem:pmbound}
  \begin{align} \label{eq:pmbound}
    \textnormal{Var}(Q^-_A) \leq \textnormal{Var}(Q^+_A).
  \end{align}
\end{lemma}
In the earlier example of the classical fixed-charge ensemble, the charge variance originates entirely from the classical mixing between different local charge configurations. 
As noted in Table~\ref{tab:comparison_examples}, Wigner-Yanase skew information vanishes for them since $[\rho, Q_A]\,{=}\,0$. On the other hand, for pure states, $I_{\rm WY}\,{=}\,\textrm{Var}(Q_A)$, which indicates that $\textrm{Var}(Q_A^-)\,{=}\,\textrm{Var}(Q_A^+)\,{=}\,2 \textrm{Var}(Q_A)$.

% \vspace{5pt} \emph{Discussion}. 
 
\section{Discussion}

In this work, we identified static conditions under which SWSSB entails charge scrambling. The examples show that neither direction is automatic, as expected from the linearity of fluctuation diagnostics and the nonlinearity of SWSSB diagnostics. Extensive charge variance alone does not guarantee SWSSB, because it can arise from classical sampling over block charges without local charge-transfer connectivity in the support. Conversely, SWSSB alone does not force extensive variance unless the R\'enyi-1 correlations approach their nonzero asymptotic value sufficiently rapidly. This characterization should be useful for experimental studies of SWSSB, including recent cold-atom implementations~\cite{wang2026observationstrongtoweakspontaneoussymmetry}.

Furthermore, this provides an important step towards understanding the relation between SWSSB and emergent hydrodynamics in open quantum systems. In a lattice system evolved under a local strongly symmetric Lindbladian, the low-energy charge spectrum is controlled not only by SWSSB but also by the small-momentum static charge structure factor~\cite{LeeU2023}. Field-theory descriptions often phrase the same physics in terms of a susceptibility, which depends on additional dynamical data~\cite{Akyuz_2024,luca2025,Huang2025}. Our results isolate a purely static sufficient condition for the structure-factor ingredient: if a steady state has long-range R\'enyi-1 order whose approach to the nonzero asymptotic value is sufficiently rapid, then it has extensive block-charge variance and hence, under the usual regularity assumptions, a finite nonzero $k \rightarrow 0^+$ structure factor. A natural next question is what dynamical assumptions on the Lindbladian enforce this rapid convergence in the steady state.

Finally, we remark that this result also elucidates the connection between SWSSB and charge sharpening transition in monitored circuits~\cite{PhysRevLett.129.200602, PhysRevX.12.041002}, where local charge indefiniteness plays an important role~\cite{vijay2025holographicallyemergentgaugetheory,Singh_2026}. In particular, identifying the fluctuation of weak symmetry twist overlap with the Wigner-Yanase skew information provides a precise way to separate the classical and quantum contributions to charge scrambling, which can be used to understand the landscape of monitored random circuits.

% Understanding how these quantities behave dynamically, especially near monitored critical points or in systems with multiple competing transfer channels, appears to be a natural next step.
% ALso how are they related to generic operator growth or quantum chaos, as in OTOC-based diagnostics?

\acknowledgments
This paper is dedicated to Dho Eun Kum, with deep gratitude for her constant support throughout this work.
J.Y.L. acknowledges support by the faculty startup grant at the University of Illinois, Urbana-Champaign, and the IBM-Illinois Discovery Accelerator Institute. 

\bibliography{ref}

\clearpage 
\newpage
\onecolumngrid 

\appendix

\tableofcontents

\section{Doubled Hilbert space} \label{app:doubled}

For simplicity, here we assume $G$ is a finite Abelian symmetry group with unitary representation $U_g$ on a Hilbert space $\cH$.  Since $G$  is Abelian, $\cH$ decomposes into charge sectors
\begin{align}
\mathcal{H}
=
\bigoplus_{q \in \widehat{G}} \mathcal{H}_q ,
\end{align}
where $\widehat{G}$ is the character group. A state $|\psi_q \rangle$ with charge $q$ transforms as
\begin{align}
U_g |\psi_q\rangle
=
\chi_q(g) |\psi_q\rangle .
\end{align}
where $\chi_q(g)$ is the character of the representation. For an Abelian group, the irreducible representations are one-dimensional, so $\chi_q(g)$ is a complex number with unit magnitude.

The canonical purification of a density matrix $\rho$ can be regarded as a state in $\mathcal{H}_L \otimes \overline{\mathcal{H}}_R$, where $\overline{\mathcal{H}}$ denotes the conjugate Hilbert space. Explicitly,
\begin{align}
  \Vert \sqrt{\rho} \rAngle := \sum_{i} ( \sqrt{\rho} \otimes \bm{1}) \, |i \rangle_L  |\bar{i} \rangle_R = \sum_i ( \bm{1} \otimes \sqrt{\rho}^T) \, |i \rangle_L  |\bar{i} \rangle_R.
\end{align}
The bar notation in the second copy is used to indicate that the second copy carries the conjugate wavefunction. This is important as $\sum_i |i\rangle |i \rangle$ is not invariant under the basis change $|i\rangle \to V |i\rangle$, while $\sum_i |i\rangle |\bar{i} \rangle$ is invariant:
\begin{align}
  \sum_i V|i \rangle \otimes V^* |\bar{i} \rangle = \sum_{ijj'} V_{ji} |j \rangle \otimes V^*_{j'i}|\bar{j}' \rangle = \sum_{ijj'} V^\vdagger_{ji} |j \rangle \otimes V^\dagger_{ij'}|\bar{j}' \rangle = \sum_j |j \rangle \otimes |\bar{j} \rangle.
\end{align}
More generally, note that $\Vert \bm{1} \rAngle := \sum_i |i \rangle_L |\bar{i} \rangle_R$ satisfies $(O \otimes \bm{1}) \Vert \bm{1} \rAngle = (\bm{1} \otimes O^T) \Vert \bm{1} \rAngle$ for any operator $O$.

If we consider a weak symmetry action $\rho \mapsto U_g^\vdagger \rho U_g^\dagger$, we see that
\begin{align}
  \Vert U_g^\vdagger \sqrt{\rho}U_g^\dagger \rAngle = (U_g^{\vphantom{*}} \otimes U_g^*) \Vert \sqrt{\rho} \rAngle.
\end{align}
Therefore, if we consider a state $| q \rangle \langle q | \mapsto |q \rangle | \bar{q} \rangle$ in the doubled Hilbert space, the state is invariant under the weak symmetry action:
\begin{align}
  U_g \otimes U_g^* | q \rangle | \bar{q} \rangle = \chi_q(g) \chi^*_{q} (g) | q \rangle | \bar{q} \rangle = | q \rangle | \bar{q} \rangle.
\end{align}

When the state has a strong symmetry, the canonically purified state has a well-defined left and right symmetry, with the weak symmetry being the ``diagonal'' subgroup given as $\{ U_g \otimes U_g^* : g \in G \}$.
The enlarged doubled symmetry may be denoted schematically by $G_L \times G_R$.
Thus, in the doubled Hilbert space, the distinction between weak and strong symmetry becomes the distinction between fixing only the diagonal charge and fixing both the diagonal and relative charges.

\section{Additional Proofs} \label{app:proofs}

We use a basis-independent convention for right-copy operators. 
The right Hilbert space is the conjugate Hilbert space 
$\overline{\mathcal H}$. For an operator $O$ on $\mathcal H$, 
let $\overline O$ denote the corresponding operator on 
$\overline{\mathcal H}$, defined by
\begin{align}
  \overline O\,\overline{|\psi\rangle} = \overline{O|\psi\rangle}.
\end{align}
Equivalently, if $O$ is represented by a matrix $O_{ij}$ in some basis, 
then $\overline O$ is represented by $O^*$ on the conjugate basis.

For a Hermitian charge operator $Q_A$, we define
\begin{align}
  Q_A^L := Q_A\otimes \mathbf 1,\qquad
Q_A^R := \mathbf 1\otimes \overline{Q}_A,\qquad
Q_A^\pm := Q_A^L \pm Q_A^R .
\end{align}
With this convention, if $U_\theta=e^{i\theta Q_A}$, then
$\overline{U_\theta}=e^{-i\theta\overline{Q}_A}$, so $Q_A^-$
generates the weak action $U_\theta\otimes \overline{U_\theta}$
on the canonical purification. 
In the occupation-number basis,
where $Q_A$ is diagonal, this reduces to the simpler notation
$Q_A^R=\mathbf 1\otimes Q_A$. We will repeatedly use the identity $\lAngle X \Vert A\otimes \overline B \Vert X \rAngle = \tr(X^\dagger A X B^\dagger)$.

\vspace{5pt} \noindent {\bf Proof of Lemma~\ref{lem:decomposition_var}.}
\begin{proof}
Let $\mu=\tr(\rho Q_A)$ and $\Delta Q_A=Q_A-\mu$. Since $Q_A$ is Hermitian, one has
\begin{align}
\langle\!\langle (Q_A^L)^2\rangle\!\rangle
= \langle\!\langle (Q_A^R)^2\rangle\!\rangle = \tr(\rho Q_A^2),
\end{align}
and
\begin{align}
\langle\!\langle Q_A^L Q_A^R\rangle\!\rangle
= \tr(\sqrt{\rho} Q_A \sqrt{\rho} Q_A).
\end{align}
Moreover,
\begin{align}
\langle\!\langle Q_A^L\rangle\!\rangle
= \langle\!\langle Q_A^R\rangle\!\rangle
= \mu .
\end{align}
Therefore
\begin{align}
\Var(Q_A^+)
= 2\tr(\rho Q_A^2)
+ 2\tr(\sqrt{\rho} Q_A \sqrt{\rho}Q_A)
- 4\mu^2
= 2\Var(Q_A)
+ 2\tr(\sqrt{\rho}\Delta Q_A\sqrt{\rho}\Delta Q_A),
\end{align}
while
\begin{align}
\Var(Q_A^-)
= 2\tr(\rho Q_A^2) - 2\tr(\sqrt{\rho} Q_A \sqrt{\rho}Q_A)
= 2\Var(Q_A) - 2\tr(\sqrt{\rho}\Delta Q_A\sqrt{\rho}\Delta Q_A),
\end{align}
where we used $\langle\!\langle Q_A^-\rangle\!\rangle=0$.
Adding the two identities gives
\begin{align}
\Var(Q_A^+)+\Var(Q_A^-)=4\Var(Q_A),
\end{align}
or equivalently
\begin{align}
\Var(Q_A)=\frac14\Var(Q_A^+)+\frac14\Var(Q_A^-).
\end{align}
\end{proof}

% \begin{proof}
%     To show this, expand 
%     \begin{align}
%     & \textnormal{Var}(Q_A^+) = \lAngle (Q_A^L + Q_A^R)^2 \rAngle -  \lAngle Q_A^L + Q_A^R \rAngle^2 \nonumber \\
%     &= 2 \tr \rho Q_A^2 + 2 \tr \sqrt{\rho} Q_A \sqrt{\rho} Q_A - 4 ( \tr \rho Q_A )^2 \nonumber \\
%     &= 4 \textnormal{Var}(Q_A) - 2 \big( \tr \rho Q_A^2 - \tr \sqrt{\rho} Q_A \sqrt{\rho} Q_A \big) \nonumber \\
%     &= 4 \textnormal{Var}(Q_A) - \lAngle  (Q_A^-)^2 \rAngle = 4 \textnormal{Var}(Q_A) - \textnormal{Var}(Q_A^-).
%     \end{align}
%   In the last line, we used $\lAngle Q_A^- \rAngle = \lAngle Q_A^L - Q_A^R \rAngle = 0$.
% \end{proof}

\vspace{5pt} \noindent {\bf Proof of Lemma~\ref{lem:pmbound}.}
\begin{proof}
  Let $\mu=\tr(\rho Q_A)$ and $\Delta Q_A=Q_A-\mu$. Then
  \begin{align}
    \textnormal{Var}(Q_A^-) &= 2 \tr\rho Q_A^2 - 2 \tr \sqrt{\rho} Q_A \sqrt{\rho} Q_A \nonumber \\
    &= 2\tr \rho (\Delta Q_A)^2  - 2 \tr \sqrt{\rho} \Delta Q_A \sqrt{\rho} \Delta Q_A \nonumber \\
    &\leq 2 \textnormal{Var}(Q_A)
  \end{align}
  where we used $\textnormal{Var}(Q_A) = \tr \rho (\Delta Q_A)^2$. Note that
  \begin{align}
    \tr \big( \sqrt{\rho} \Delta Q_A \sqrt{\rho} \Delta Q_A \big) = \tr \big( ( \rho^{1/4} \Delta Q_A \rho^{1/4} )^2 \big) \geq 0,
  \end{align}
  since $\rho^{1/4} \Delta Q_A \rho^{1/4} $ is a Hermitian operator and $\tr X^\dagger X \geq 0$ for any $X$. Finally, using Lemma~\ref{lem:decomposition_var}, we obtain \eqnref{eq:pmbound}.
\end{proof}

\section{Sparse projector states}

\subsection{Vanishing R\'enyi-1 Order yet Extensive Subsystem Charge Variance}\label{app:countexample}

In this appendix we show that the converse of Theorem~\ref{thm:main} is false: extensive subsystem charge variance does not, by itself, imply SWSSB. The counterexample is a typical sparse classical projector inside a fixed-charge sector.

\vspace{5pt} \noindent {\bf Setup.} Fix a filling fraction $\nu$, and consider a sequence of systems with $V$ sites and total charge $Q=\nu V$. Let
\begin{align}
\mathcal C_Q
=
\Bigl\{
\mathbf n=(n_1,\dots,n_V):
n_r\in\{0,1\},\ \sum_{r=1}^V n_r=Q
\Bigr\},
\qquad
\mathcal H_Q=\mathrm{span}\{\ket{\mathbf n}:\mathbf n\in\mathcal C_Q\},
\end{align}
so that
\begin{align}
d_Q:=\dim \mathcal H_Q=|\mathcal C_Q|=\binom{V}{Q}.
\end{align}
Choose a subset $\Omega\subset \mathcal C_Q$ uniformly at random with cardinality
\begin{align}
|\Omega|=r(V),
\qquad
V^2\ll r(V),
\qquad
\log r(V)=o(V).
\end{align}
where $o(V)$ means subextensive, i.e., $\lim_{V\to\infty} o(V)/V=0$.
We then define the mixed state
\begin{align}
\rho_\Omega
=
\frac{1}{r(V)}
\sum_{\mathbf n\in\Omega}
\ket{\mathbf n}\bra{\mathbf n}.
\end{align}
Since every basis state in the support of $\rho_\Omega$ has the same total charge $Q$, the state is strongly $U(1)$-symmetric.

We will show that, with probability approaching one as $V \rightarrow \infty$, a typical realization $\rho_\Omega$ has the following two properties: (1) $\rho_\Omega$ has extensive subsystem charge variance. (2) $\rho_\Omega$ has no long-range R\'enyi-1 order for any fixed finite-support charge-transfer operator.
Such a typical realization $\rho_\Omega$ provides a counterexample to the converse of Theorem~\ref{thm:main}.

\vspace{5pt} \noindent {\bf Preliminary step: Uniform ensemble.} 
Let $A$ be a subregion with
\begin{align}
|A|=\ell=\alpha V,
\qquad 0<\alpha < 1,
\end{align}
and define the block charge
\begin{align}
Q_A=\sum_{r\in A} n_r.
\end{align}
We first compute the mean and variance of $Q_A$ in the uniform ensemble of fixed-charge sector $\cC_Q$. The uniform distribution on $\cC_Q$ is invariant under permutations of site labels, so all one-site marginals are identical and all two-site marginals for distinct sites are identical. Therefore
\begin{align}
\mathbb E_{\mathcal C_Q}[n_r] = \frac{Q}{V} = \nu,
\qquad
\mathbb E_{\mathcal C_Q}[n_r n_s] = \frac{Q(Q-1)}{V(V-1)}
\qquad (r\neq s).
\end{align}
It follows that
\begin{align}
\mu_A := \mathbb E_{\mathcal C_Q}[Q_A] = \ell\frac{Q}{V} = \ell\nu,
\end{align}
and
\begin{align} \label{eq:uniform_var}
\sigma_A^2 := \Var_{\mathcal C_Q}(Q_A) = m_{2,A} - \mu_A^2 = \sum_{r\in A}\Var(n_r) + \sum_{\substack{r,s\in A\\r\neq s}} \Cov(n_r,n_s) =\ell\frac{Q}{V}\biggl(1-\frac{Q}{V}\biggr)\frac{V-\ell}{V-1},
\end{align}
where $m_{2,A}=\mathbb E_{\mathcal C_Q}[Q_A^2]$ is the population second moment.
% Equivalently, the block-charge distribution is hypergeometric,
% \begin{align}
% p_m
% =
% \frac{\binom{\ell}{m}\binom{V-\ell}{Q-m}}{\binom{V}{Q}},
% \end{align}
% with the same variance $\sigma_A^2$. 
In particular, for $\ell=\alpha V$,
\begin{align}
\sigma_A^2
=
\alpha(1-\alpha)\nu(1-\nu)V+O(1),
\end{align}
so the full fixed-charge sector has extensive block-charge variance.

\vspace{5pt} \noindent {\bf Step I: Extensive subsystem charge variance}.  We want to show that a typical sparse subset $\Omega$ inherits the same extensive scaling. For each $\mathbf n\in\mathcal C_Q$, define
\begin{align}
x_{\mathbf n}:=Q_A(\mathbf n),
\qquad
I_{\mathbf n}:=
\begin{cases}
1,& \mathbf n\in\Omega,\\[3pt]
0,& \mathbf n\notin\Omega.
\end{cases}
\end{align}
Then
\begin{align}
\Tr(\rho_\Omega Q_A)
=
\frac{1}{r}
\sum_{\mathbf n\in\mathcal C_Q}
I_{\mathbf n}x_{\mathbf n},
\qquad
\Tr(\rho_\Omega Q_A^2)
=
\frac{1}{r}
\sum_{\mathbf n\in\mathcal C_Q}
I_{\mathbf n}x_{\mathbf n}^2,
\end{align}
where for brevity we write $r=r(V)$. Since $\Omega$ is a uniformly random subset of size $r$, one has
\begin{align}
\mathbb E_\Omega[I_{\mathbf n}]
=
\frac{r}{d_Q},
\qquad
\mathbb E_\Omega[I_{\mathbf n}I_{\mathbf m}]
=
\frac{r(r-1)}{d_Q(d_Q-1)}
\qquad
(\mathbf n\neq \mathbf m).
\end{align}
The expectation is immediate:
\begin{align}
\mathbb E_\Omega[\Tr(\rho_\Omega Q_A)] = \frac{1}{r}\sum_{\mathbf n\in\mathcal C_Q} \mathbb E_\Omega[I_{\mathbf n}]\,x_{\mathbf n} =
\frac{1}{r}\cdot \frac{r}{d_Q}\sum_{\mathbf n\in\mathcal C_Q} x_{\mathbf n} = \mu_A.
\end{align}
Similarly for the variance, where $y_{\mathbf n}=x_{\mathbf n}-\mu_A$ are the centered values,
\begin{align}
\Var_\Omega(\Tr(\rho_\Omega Q_A)) =
\frac{1}{r^2}
\mathbb E_\Omega\Biggl[
\Bigl(\sum_{\mathbf n} I_{\mathbf n} y_{\mathbf n}\Bigr)^2
\Biggr] 
= 
\frac{1}{r^2}
\Biggl[
\sum_{\mathbf n}\mathbb E_\Omega[I_{\mathbf n}]\,y_{\mathbf n}^2
+
\sum_{\mathbf n\neq \mathbf m}
\mathbb E_\Omega[I_{\mathbf n}I_{\mathbf m}]\,y_{\mathbf n}y_{\mathbf m}
\Biggr]
% = \frac{1}{r^2} \Biggl[ \frac{r}{d_Q}\sum_{\mathbf n} y_{\mathbf n}^2 + \frac{r(r-1)}{d_Q(d_Q-1)} \sum_{\mathbf n\neq \mathbf m} y_{\mathbf n}y_{\mathbf m} \Biggr].
\end{align}
Since $\sum_{\mathbf n} y_{\mathbf n}=0$, we have $\sum_{\mathbf n\neq \mathbf m} y_{\mathbf n}y_{\mathbf m}=-\sum_{\mathbf n} y_{\mathbf n}^2$. Hence
\begin{align}
\Var_\Omega(\Tr(\rho_\Omega Q_A))
&= \frac{1}{r^2}
\Biggl[
\frac{r}{d_Q} -\frac{r(r-1)}{d_Q(d_Q-1)}
\Biggr]
\sum_{\mathbf n} y_{\mathbf n}^2
= \frac{d_Q-r}{rd_Q(d_Q-1)} \sum_{\mathbf n} y_{\mathbf n}^2 = \frac{d_Q-r}{r(d_Q-1)}\,\sigma_A^2.
\end{align}
% From the standard finite-population sampling formulas,
% \begin{align}
% \mathbb E_\Omega\bigl[\Tr(\rho_\Omega Q_A)\bigr]
% =
% \mu_A,
% \qquad
% \Var_\Omega\bigl(\Tr(\rho_\Omega Q_A)\bigr)
% =
% \frac{d_Q-r}{r(d_Q-1)}\,\sigma_A^2.
% \end{align}
Likewise, if we consider a second moment,
% $m_{2,A}:=\mathbb E_{\mathcal C_Q}[Q_A^2]$, 
then
\begin{align}
\mathbb E_\Omega\bigl[\Tr(\rho_\Omega Q_A^2)\bigr]
=
\mathbb E_{\mathcal C_Q}[Q_A^2] =: m_{2,A},
\qquad
\Var_\Omega\bigl(\Tr(\rho_\Omega Q_A^2)\bigr)
=
\frac{d_Q-r}{r(d_Q-1)}\Var_{\mathcal C_Q}(Q_A^2).
\end{align}
which are all written in terms of the expectation values for the uniform distribution.

For a specific realization $\rho_\Omega$, define the subsystem charge variance
\begin{align}
V_A(\Omega)
:=
\Var_{\rho_\Omega}(Q_A)
=
\Tr(\rho_\Omega Q_A^2)-\Tr(\rho_\Omega Q_A)^2.
\end{align}
Since $\sigma_A^2 = m_{2,A}-\mu_A^2$, we have
\begin{align}
|V_A(\Omega)-\sigma_A^2|
\le
\bigl|\Tr(\rho_\Omega Q_A^2)-m_{2,A}\bigr|
+
\bigl|\Tr(\rho_\Omega Q_A)^2-\mu_A^2\bigr|.
\end{align}
In order to bound the probability for $V_A(\Omega)$ to deviate from $\sigma_A^2$ by $\epsilon V$, we use the union bound. By the union bound, 
\begin{align}
& X_\Omega+Y_\Omega \geq a
\subset
\{ X_\Omega \geq \frac{a}{2} \}
\cup
\{Y_\Omega\ge \frac{a}{2}\} \nonumber \\
\Rightarrow & \Pr_\Omega(X_\Omega+Y_\Omega\ge a)
\le
\Pr_\Omega\left(X_\Omega\ge \frac{a}{2}\right)
+
\Pr_\Omega\left(Y_\Omega\ge \frac{a}{2}\right).
\end{align}
where $X_\Omega = | \Tr( \rho_\Omega Q_A^2) - m_{2,A} |$ and $Y_\Omega = |\Tr(\rho_\Omega Q_A)^2-\mu_A^2|$.
Therefore, for any $\varepsilon>0$,
\begin{align}
\Pr_\Omega\left(|V_A(\Omega)-\sigma_A^2|\ge \varepsilon V\right)
&\le
\Pr_\Omega\left(X_\Omega+Y_\Omega\ge \varepsilon V\right)
\le
\Pr_\Omega\left(X_\Omega\ge \frac{\varepsilon V}{2}\right)
+
\Pr_\Omega\left(Y_\Omega\ge \frac{\varepsilon V}{2}\right),
\end{align} 
Since both $\Tr(\rho_\Omega Q_A)$ and $\mu_A$ lie in $[0,\ell]$,
\begin{align}
\bigl|\Tr(\rho_\Omega Q_A)^2-\mu_A^2\bigr|
\le
2\ell\bigl|\Tr(\rho_\Omega Q_A)-\mu_A\bigr|.
\end{align}
Recall that $\ell = \alpha V$. Using the above relation, we can further simplify the inequality as
\begin{align}
\Pr_\Omega\bigl(|V_A(\Omega)-\sigma_A^2|\ge \varepsilon V\bigr)
&\le
\Pr_\Omega\bigl(|\Tr(\rho_\Omega Q_A^2)-m_{2,A}|\ge \varepsilon V/2 \bigr) +
\Pr_\Omega\bigl(|\Tr(\rho_\Omega Q_A)-\mu_A|\ge \varepsilon/(4\alpha)\bigr).
\end{align}
Since $0\le Q_A\le \ell$, we have $0\le Q_A^2\le \ell^2$, and hence $\Var_{\mathcal C_Q}(Q_A^2)\le \ell^4$, where $\ell = \alpha V$. Using this crude bound with  Chebyshev's inequality, we obtain
\begin{align}
\Pr_\Omega\bigl(|\Tr(\rho_\Omega Q_A^2)-m_{2,A}|\ge \varepsilon V/2\bigr)
\le
\frac{4\ell^4}{\varepsilon^2 V^2 r}
=
O\Bigl(\frac{V^2}{r}\Bigr),
\end{align}
and using $\sigma_A^2 = O(V)$ as in \eqnref{eq:uniform_var}, we have
\begin{align}
\Pr_\Omega\bigl(|\Tr(\rho_\Omega Q_A)-\mu_A|\ge \varepsilon/(4\alpha)\bigr)
\le
\frac{16\alpha^2}{\varepsilon^2}
\frac{d_Q-r}{r(d_Q-1)}\sigma_A^2
=
O\Bigl(\frac{V}{r}\Bigr).
\end{align}
Because $V^2\ll r(V)$, both probabilities vanish as $V\to\infty$. Hence, with probability approaching one,
\begin{align}
V_A(\Omega)
=
\sigma_A^2+o(V)
=
\alpha(1-\alpha)\nu(1-\nu)V+o(V).
\end{align}
Thus a typical sparse projector $\rho_\Omega$ has extensive subsystem charge variance.

\vspace{5pt} \noindent {\bf Step II: Absence of R\'enyi-1 order.}
We now show that the same typical state has no long-range R\'enyi-1 order. Consider first the standard single-particle charge-transfer operator
\begin{align}
O_{ij}=c_i^\dagger c_j,
\qquad
i\neq j.
\end{align}
Since $\rho_\Omega$ is diagonal in the occupation basis, there is no ordinary off-diagonal long-range order:
\begin{align}
C_{ij} := \Tr(\rho_\Omega O_{ij})=0.
\end{align}
The relevant nonlinear diagnostic is the R\'enyi-1 correlator
\begin{align}
R^{(\Omega)}_{ij}
=
\Tr\bigl(
O_{ij}\rho_\Omega^{1/2}O_{ij}^\dagger \rho_\Omega^{1/2}
\bigr) =
\frac{1}{r}
\Tr\bigl(
c_i^\dagger c_j
\Pi_\Omega
c_j^\dagger c_i
\Pi_\Omega
\bigr).
\end{align}
where $\Pi_\Omega = \sum_{\mathbf n \in \Omega} | \mathbf{n} \rangle \langle \mathbf{n} |$ thus $\rho^{1/2} = \frac{1}{\sqrt{r}} \Pi_\Omega$.
Let
\begin{align}
S_{ij}
=
\bigl\{
\mathbf n\in\mathcal C_Q:
n_j=1,\ n_i=0
\bigr\},
\end{align}
and for $\mathbf n\in S_{ij}$, let $T_{j\to i}(\mathbf n)=\mathbf n^{j\to i}$ denote the configuration obtained by moving one particle from $j$ to $i$. Then
\begin{align}
c_i^\dagger c_j\ket{\mathbf n}
=
\ket{T_{j\to i}(\mathbf n)}
\qquad
(\mathbf n\in S_{ij}),
\end{align}
and vanishes otherwise. Expanding the trace in the occupation basis yields
\begin{align}
R^{(\Omega)}_{ij}
=
\frac{1}{r}
\sum_{\mathbf n\in S_{ij}}
1_{\mathbf n\in\Omega}\,
1_{T_{j\to i}(\mathbf n)\in\Omega}.
\end{align}
The summation is over only $S_{ij}$ since for other basis states $c_i^\dagger c_j$ simply vanishes. 
Thus $R^{(\Omega)}_{ij}$ counts the number of directed local hopping edges
\begin{align}
\mathbf n\to T_{j\to i}(\mathbf n)
\end{align}
that remain inside $\Omega$, divided by $r$. 
Intuitively, conditioned on $\mathbf n\in\Omega$, the probability that $T_{j\to i}(\mathbf n)\in\Omega$ is
\begin{align}
\frac{r-1}{d_Q-1}\sim \frac{r}{d_Q},
\end{align}
which implies that the R\'enyi-1 correlator gets exponentially suppressed.

To compute the correlator exactly, first note that the number of allowed sources is
\begin{align}
|S_{ij}|
=
\binom{V-2}{Q-1},
\end{align}
since one fixes $n_j=1$ and $n_i=0$, and then chooses the remaining $Q-1$ occupied sites among the other $V-2$ sites. For a fixed $\mathbf n\in S_{ij}$, the two configurations $\mathbf n$ and $T_{j\to i}(\mathbf n)$ are distinct, and the probability that both belong to a uniformly random $r$-element subset $\Omega\subset\mathcal C_Q$ is
\begin{align}
\Pr_\Omega\bigl(
\mathbf n\in\Omega
\text{ and }
T_{j\to i}(\mathbf n)\in\Omega
\bigr)
=
\frac{\binom{d_Q-2}{r-2}}{\binom{d_Q}{r}}
=
\frac{r(r-1)}{d_Q(d_Q-1)}.
\end{align}
Therefore,
\begin{align}
\mathbb E_\Omega\bigl[R^{(\Omega)}_{ij}\bigr]
&=
\frac{1}{r}
|S_{ij}|
\frac{r(r-1)}{d_Q(d_Q-1)}
\nonumber\\
&=
\frac{|S_{ij}|}{d_Q}\frac{r-1}{d_Q-1}.
\end{align}
Using
\begin{align}
\frac{|S_{ij}|}{d_Q}
=
\frac{\binom{V-2}{Q-1}}{\binom{V}{Q}}
=
\frac{Q(V-Q)}{V(V-1)},
\end{align}
we obtain the exact formula
\begin{align}
\mathbb E_\Omega\bigl[R^{(\Omega)}_{ij}\bigr]
=
\frac{Q(V-Q)}{V(V-1)}
\frac{r-1}{d_Q-1}.
\end{align}
Since
\begin{align}
d_Q
=
\binom{V}{Q}
=
\exp\bigl[Vh(\nu)+o(V)\bigr],
\qquad
h(\nu)=-\nu\log\nu-(1-\nu)\log(1-\nu),
\end{align}
and $\log r(V)=o(V)$, we conclude that
\begin{align}
\mathbb E_\Omega\bigl[R^{(\Omega)}_{ij}\bigr]
\sim
\nu(1-\nu)\frac{r}{d_Q}
=
\exp\bigl[-Vh(\nu)+o(V)\bigr].
\end{align}
Thus the ensemble average of the R\'enyi-1 hopping correlator is exponentially small for any pair $(i,j)$.

To show that this is in fact a typical behavior, we define the raw edge count
\begin{align}
Y_{ij}(\Omega)
=
r R^{(\Omega)}_{ij}
=
\sum_{\mathbf n\in S_{ij}}
1_{\mathbf n\in\Omega}\,
1_{T_{j\to i}(\mathbf n)\in\Omega} \in \mathbb{Z}.
\end{align}
A nonzero value of $Y_{ij}$ detects nonzero R\'enyi-1 correlator. From the previous step, it is straightforward that
\begin{align}
\mathbb E_\Omega\bigl[Y_{ij}(\Omega)\bigr]
=
|S_{ij}|
\frac{r(r-1)}{d_Q(d_Q-1)}
\sim
\nu(1-\nu)\frac{r^2}{d_Q}.
\end{align}
Since $Y_{ij}(\Omega)$ is a nonnegative integer and $r \sim \exp[o(V)]$, we obtain
\begin{align}
\Pr_\Omega\bigl(R^{(\Omega)}_{ij}>0\bigr)
=
\Pr_\Omega\bigl(Y_{ij}(\Omega)\ge 1\bigr)
\le
\mathbb E_\Omega\bigl[Y_{ij}(\Omega)\bigr]
=
\exp\bigl[-Vh(\nu)+o(V)\bigr].
\end{align}
A union bound over all $O(V^2)$ ordered pairs $(i,j)$ with $i\neq j$ gives
\begin{align}
\Pr_\Omega\bigl(
\exists\, i\neq j:\ R^{(\Omega)}_{ij}>0
\bigr)
\le
V(V-1)\exp\bigl[-Vh(\nu)+o(V)\bigr]
\to 0.
\end{align}
Therefore, with probability approaching one,
\begin{align}
R^{(\Omega)}_{ij}=0
\qquad
\text{for all }i\neq j.
\end{align}
In particular, a typical sparse projector has no long-range R\'enyi-1 order.

\vspace{5pt} \noindent {\bf Step III: Extension to general finite-support charge-transfer operators.}
The same mechanism kills not only the single-particle hopping operator $c_i^\dagger c_j$, but any fixed finite-support charge-transfer operator. Let $O$ be such an operator, with support size bounded independently of $V$, and suppose that $O$ has no diagonal matrix elements in the occupation basis. Then
\begin{align}
R_1^{(\Omega)}[O]
&=
\Tr\bigl(
O\sqrt{\rho_\Omega}O^\dagger\sqrt{\rho_\Omega}
\bigr) =
\frac{1}{r}
\sum_{\mathbf n,\mathbf m\in\Omega}
\bigl|
\bra{\mathbf m}O\ket{\mathbf n}
\bigr|^2.
\end{align}
Define the directed edge set
\begin{align}
E_O
=
\bigl\{
(\mathbf n,\mathbf m)\in\mathcal C_Q\times\mathcal C_Q:
\bra{\mathbf m}O\ket{\mathbf n}\neq 0
\bigr\}.
\end{align}
Because $O$ acts on only $O(1)$ sites, each source configuration $\mathbf n$ has at most $K_O=O(1)$ targets $\mathbf m$ with nonzero matrix element, so
\begin{align}
|E_O|\le K_O d_Q
\end{align}
for some constant $K_O$ depending only on the operator $O$, but not on $V$. Now $R_1^{(\Omega)}[O]>0$ can occur only if there exists at least one edge $(\mathbf n,\mathbf m)\in E_O$ such that both $\mathbf n,\mathbf m\in\Omega$. Hence
\begin{align}
\Pr_\Omega\bigl(R_1^{(\Omega)}[O]>0\bigr)
&\le
\sum_{(\mathbf n,\mathbf m)\in E_O}
\Pr_\Omega(\mathbf n,\mathbf m\in\Omega)
=
|E_O|\frac{r(r-1)}{d_Q(d_Q-1)}
\le
K_O\frac{r(r-1)}{d_Q-1}.
\end{align}
Since $\log r(V)=o(V)$ but $d_Q$ is exponential in $V$, the right-hand side vanishes as $V\to\infty$. Therefore, with probability approaching one,
\begin{align}
R_1^{(\Omega)}[O]=0
\end{align}
for every fixed finite-support charge-transfer operator $O$.

We conclude that a typical sparse projector $\rho_\Omega$ is strongly $U(1)$-symmetric and has extensive subsystem charge variance, but exhibits neither long-range R\'enyi-1 order nor SWSSB. This provides the desired counterexample to the converse of Theorem~\ref{thm:main}.

\subsection{Conditional mutual information: sparse projectors versus the uniform ensemble.} \label{app:CMI}

Another important diagnostic of SWSSB is the von Neumann conditional mutual information (CMI) $I(C:E|B)$ for a tripartition of the system into a center region $C$, a shielding region $B$, and an environment region $E$:
\[
\centering
    \includegraphics[width=0.2\linewidth]{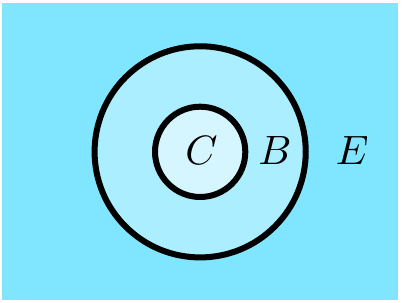}
\]
Following this convention, we partition the system into three regions
\begin{align}
C\sqcup B\sqcup E=\{1,\dots,V\},
\qquad
|C|=c,\quad |B|=b,\quad |E|=e.
\end{align}
Since both $\rho_\Omega$ and $\rho_Q$ are diagonal in the occupation basis, their von Neumann conditional mutual information is equal to the Shannon conditional mutual information of the corresponding classical random variables. We denote
\begin{align}
I(C:E|B)
=
H(C|B)+H(E|B)-H(CE|B).
\end{align}

\vspace{5pt} \noindent {\bf Sparse projector.} For a fixed configuration in $B$ denoted as $\beta \in \{0,1\}^B$, define the set of configurations in $\Omega$ that restrict to $\beta$ on $B$ (``fiber'') as
\begin{align}
\Omega_\beta
:=
\{\mathbf n\in\Omega:\mathbf n|_B=\beta\},
\qquad
K_\beta:=|\Omega_\beta|.
\end{align}
Since $\rho_\Omega$ is classical and uniform on $\Omega$, conditioning on $B=\beta$ gives a uniform distribution on the $K_\beta$ configurations in $\Omega_\beta$. Hence
\begin{align}
H_{\rho_\Omega}(CE\,|\,B=\beta)=\log K_\beta.
\end{align}

Note that
\begin{align}
  H(X) - H(X|Y) = I(X:Y) \leq H(X) \leq H(XY) \quad \Rightarrow \quad I(X:Y) \leq H(XY).
\end{align}
Therefore
\begin{align}
I_{\rho_\Omega}(C:E\,|\,B=\beta)
\le
H_{\rho_\Omega}(CE\,|\,B=\beta)
=
\log K_\beta.
\end{align}
Averaging over $\beta$, and noting that $\Pr_{\rho_\Omega}(B=\beta)=K_\beta/r$, yields
\begin{align}
I_{\rho_\Omega}(C:E\,|\,B)
\le
\frac{1}{r}\sum_{\beta} K_\beta \log K_\beta.
\end{align}
Using the elementary bound $\log K\le K-1$ for $K\ge 1$, we obtain
\begin{align}
I_{\rho_\Omega}(C:E\,|\,B)
\le
\frac{1}{r}\sum_{\beta} K_\beta(K_\beta-1).
\end{align}

The quantity $\sum_\beta K_\beta(K_\beta-1)$ counts the number of ordered pairs of distinct sampled configurations in $\Omega$ that have the same restriction to $B$. Taking the expectation over the random choice of $\Omega$, we find
\begin{align}
\mathbb E_\Omega\! \bigg[\sum_{\beta} K_\beta(K_\beta-1)\bigg]
=
r(V)\bigl(r(V)-1\bigr)\,p_{\mathrm{coll}}^{\neq}(B),
\end{align}
where $p_{\mathrm{coll}}^{\neq}(B)$ is the probability that two distinct uniformly random configurations from the full fixed-charge sector agree on the entire boundary region $B$:
\begin{align}
p_{\mathrm{coll}}^{\neq}(B)
:=
\frac{1}{d_Q(d_Q-1)}
\sum_{\beta} N_\beta(N_\beta-1),
\qquad
N_\beta
:=
\#\{\mathbf n\in\mathcal C_Q:\mathbf n|_B=\beta\} =\binom{V-b}{Q-|\beta|}.
\end{align}
Thus
\begin{align}
\mathbb E_\Omega I_{\rho_\Omega}(C:E\,|\,B)
\le
\bigl(r(V)-1\bigr)\,p_{\mathrm{coll}}^{\neq}(B).
\end{align}

If $b=o(V)$ and the filling $0<\nu<1$ is fixed, the fixed-charge constraint only changes the Bernoulli collision probability at subexponential order in $b$. More precisely, at the exponential-rate level needed below,
\begin{align}
p_{\rm coll}^{\neq}(B)
&=\kappa_\nu^b\exp[o(b)],
\qquad 
\kappa_\nu=\nu^2+(1-\nu)^2<1.
\end{align}
Equivalently,
\begin{align}
\frac{1}{b}\log p_{\rm coll}^{\neq}(B)
&\to \log\kappa_\nu
\end{align}
for any sequence $b=o(V)$ with $b\to\infty$. Hence
\begin{align}
\mathbb{E}_\Omega I_{\rho_\Omega}(C:E|B)
&\le [r(V)-1]\kappa_\nu^b\exp[o(b)].
\end{align}
Since the CMI is nonnegative, Markov's inequality gives
\begin{align}
\Pr_\Omega\left(I_{\rho_\Omega}(C:E|B)>\epsilon\right)
&\le \frac{1}{\epsilon}\mathbb{E}_\Omega I_{\rho_\Omega}(C:E|B).
\end{align}
Therefore, if
\begin{align}
\log r(V)+b\log\kappa_\nu+o(b)
&\to -\infty,
\end{align}
then
\begin{align}
I_{\rho_\Omega}(C:E|B)
&\to 0
\end{align}
in probability. In particular, because $\log r(V)=o(V)$, it is sufficient to choose a shielding region satisfying
\begin{align}
b=o(V),
\qquad 
b\gg\log r(V).
\end{align}
For example, if $r(V)\le V^p$ for some fixed $p$, then $b=C\log V$ with any constant $C>p/|\log\kappa_\nu|$ suffices to make the CMI vanish.

\vspace{5pt} \noindent {\bf Uniform fixed-charge ensemble.}
We now contrast this with the uniform fixed-charge ensemble $\rho_Q$. Fix a boundary pattern $\beta\in\{0,1\}^B$, and let
\begin{align}
m:=|\beta|=Q_B(\beta).
\end{align}
Conditioned on $B=\beta$, the configurations on $C\cup E$ are distributed uniformly subject only to the residual charge constraint
\begin{align}
Q_C+Q_E=Q-m.
\end{align}
Equivalently,
\begin{align}
\Pr_{\rho_Q}(n_C,n_E\,|\,B=\beta)
=
\frac{
\mathbf 1\{|n_C|+|n_E|=Q-m\}
}{
\binom{c+e}{Q-m}
}.
\end{align}
If we first condition on $Q_C=q$, then the configurations on $C$ and $E$ are independent and uniformly distributed in their respective charge sectors:
\begin{align}
\Pr_{\rho_Q}(n_C,n_E\,|\,B=\beta)
=
\sum_q
\Pr_{\rho_Q}(Q_C=q\,|\,B=\beta)\,
\Pr(n_C\,|\,Q_C=q)\,
\Pr(n_E\,|\,Q_E=Q-m-q),
\end{align}
where
\begin{align}
\Pr_{\rho_Q}(Q_C=q\,|\,B=\beta)
=
\frac{\binom{c}{q}\binom{e}{Q-m-q}}{\binom{c+e}{Q-m}}.
\end{align}
This is a hypergeometric distribution.

From this decomposition one finds the exact identity
\begin{align}
I_{\rho_Q}(C:E\,|\,B=\beta)
=
H_{\rho_Q}(Q_C\,|\,B=\beta),
\end{align}
and hence
\begin{align}
I_{\rho_Q}(C:E\,|\,B)
=
H_{\rho_Q}(Q_C\,|\,B).
\end{align}
So the conditional mutual information of the uniform fixed-charge ensemble is exactly the conditional entropy of the charge in region $C$ after conditioning on $B$.

This immediately shows that the CMI of $\rho_Q$ remains nonzero. Indeed, if $c=|C|=O(1)$ is fixed and $b=o(V)$, then the hypergeometric law converges to a binomial law,
\begin{align}
Q_C\,|\,B
\ \Longrightarrow\
\mathrm{Bin}(c,\nu),
\end{align}
and therefore
\begin{align}
I_{\rho_Q}(C:E\,|\,B)
=
H\!\left(\mathrm{Bin}(c,\nu)\right)+o(1).
\end{align}
For a single-site region $C$, this reduces to
\begin{align}
I_{\rho_Q}(C:E\,|\,B)
=
h_2(\nu)+o(1),
\end{align}
where
\begin{align}
h_2(\nu):=-\nu\log\nu-(1-\nu)\log(1-\nu)
\end{align}
is the binary entropy. Thus the CMI of the uniform fixed-charge ensemble remains finite and strictly positive even when the shielding region $B$ is large, provided $b=o(V)$.

More generally, if both $c=|C|$ and $e=|E|$ scale linearly with $V$ while $b=o(V)$, then $Q_C|B$ has variance $\Theta(V)$, so its hypergeometric entropy grows as
\begin{align}
I_{\rho_Q}(C:E\,|\,B)
=
\frac{1}{2}\log V+O(1).
\end{align}
Hence the uniform fixed-charge ensemble retains a parametrically large CMI, in sharp contrast with the sparse projector.

\vspace{5pt} \noindent {\bf Summary.} For the sparse-projector state $\rho_\Omega$, the conditional mutual information is controlled by the total weight of ambiguous boundary fibers, and vanishes once $b\gg \log r(V)$. By contrast, for the uniform fixed-charge ensemble $\rho_Q$, the CMI is exactly
\begin{align}
I_{\rho_Q}(C:E\,|\,B)=H_{\rho_Q}(Q_C\,|\,B),
\end{align}
which remains nonzero, and can even grow logarithmically for macroscopic regions. Thus the sparse-projector counterexample destroys the long-range information-theoretic correlations present in the uniform ensemble, even though both states live in the same global charge sector.

\section{Structure Factor} \label{app:structure_factor}

In this section, we define the structure factor and show how it is connected to the charge variance. Let $n(x,t)$ be a conserved charge density. To discuss $n(x,t)$, one has to specify some underlying dynamics that generates $n(x,t)$ from $n(x)$. Define the fluctuation
\begin{align}
    \delta n(x,t)
    =
    n(x,t)-\langle n\rangle.
\end{align}
In a translationally invariant system, define the Fourier mode
\begin{align}
    \delta n_k(t)
    =
    \int \dd^d x \, e^{-\ii k\cdot x}\delta n(x,t).
\end{align}

\vspace{5pt} \noindent {\bf Structure factor}. The dynamic structure factor is the nonsymmetrized correlation function
\begin{align}
    S(k,\omega)
    =
    \frac{1}{V}
    \int_{-\infty}^{\infty}\dd t \, e^{\ii \omega t}
    \left\langle
    \delta n_k(t)\delta n_{-k}(0)
    \right\rangle.
\end{align}
It measures the fluctuation spectrum of the charge density. It is not itself a response function.

The static (equal-time) structure factor is
\begin{align}
    S(k)  =
    \frac{1}{V}
    \left\langle
    \delta n_k(0) \delta n_{-k}(0)
    \right\rangle.
\end{align}
Physically, $S(k)$ measures equal-time density fluctuations at wavelength $1/k$. For example, it controls the scaling of charge fluctuations in a large region. To see this, consider the charge fluctuation in a block $A$, which is a hypercube with linear size $l$ so that $|A| = l^d$:
\begin{align}
\delta Q^{(A)} := \sum_{x \in A} \delta n_x = \sum_{x \in A} \big( n_x - \langle n_x \rangle \big).
\end{align}
To proceed, introduce the block window function
\begin{align}
w_A(x) :=
\begin{cases}
1,& x \in A,\\
0,& \text{otherwise},
\end{cases}
\qquad
\tilde w_A(k) := \sum_x w_A(x) e^{-ik \cdot x}
= \prod_j  \qty[ e^{-ik_j (l+1)/2}\,\frac{\sin(k_j l/2)}{\sin(k_j/2)} ],
\end{align}
where $A$ is defined for coordinates $1 \leq x_i \leq l$. Define $\delta Q_k := \sum_x \delta n_x e^{-ik \cdot x}$. Then 
\begin{align}
    \delta Q^{(A)} = \sum_x w_A(x) \delta n_x = \frac{1}{V}\sum_k \tilde w_A(k) \delta n_{-k}, \quad  (\delta Q^{(A)})^2 = \frac{1}{V^2} \sum_{k,q} \tilde w_A(k) \tilde w_A(q) \delta n_{-k}  \delta n_{-q},
\end{align}
where $V$ is the volume of the entire system. For a translationally invariant state, $\langle \delta n_k \delta n_q \rangle =: V \delta_{k+q,0} S(k)$. Assuming translational invariance, we obtain that
\begin{align}
\mathrm{Var}(Q^{(A)})
= \frac{1}{V}\sum_k |\tilde w_A(k)|^2 S(k).
\end{align}
Thus the block variance at scale $l$ is a weighted average of $S(k)$ over momenta with
weight $|\tilde w_A(k)|^2$. The window function $|\widetilde w_A(k)|^2$ is sharply peaked near $k=0$, has peak value $|\widetilde w_A(0)|^2=|A|^2=l^{2d}$, and has width of order $1/l$ in each momentum direction. It also satisfies the normalization identity
\begin{align}
    \frac{1}{V}\sum_k |\widetilde w_A(k)|^2 = |A|.
\end{align}
In a fixed-total-charge sector, the exact $k=0$ mode is absent or
has $S(0)=0$. The relevant quantity for large but finite subregions is
therefore the small nonzero momentum limit, $k\to0^+$.

\vspace{5pt} \noindent {\bf (1) $S(k)$ approaches a nonzero constant $\lim_{|k| \rightarrow 0^+} S(k) =: S_0 > 0$.}  Consider a continuum limit. If $S(k)$ is approximately constant over the momentum window $|k|\lesssim 1/l$, then the normalization condition gives
\begin{align}
    \mathrm{Var}(Q^{(A)})
    \approx
    S_0\frac{1}{V}\sum_k |\widetilde w_A(k)|^2
    =
    S_0 |A|.
\end{align}
Equivalently, in the continuum approximation,
\begin{align}
    \mathrm{Var}(Q^{(A)})
    =
    \int \frac{d^d k}{(2\pi)^d}
    S(k)|\widetilde w_A(k)|^2
    \approx
    S_0 l^d .
\end{align}
Thus a finite nonzero $k\to0^+$ structure factor gives a volume-law block-charge variance. Conversely, suppose that the thermodynamic-limit structure factor is locally bounded near $k=0$ and has a regular small-momentum limit. Then the normalized block variance converges to this limit:
\begin{align}
  \lim_{l\to\infty}
\frac{\mathrm{Var}(Q^{(A_l)})}{|A_l|}
= \lim_{|k|\to0^+}S(k).
\end{align}
Thus, under this regularity assumption, a finite positive volume-law coefficient of the block-charge variance is equivalent to a finite nonzero low-momentum static structure factor.

\vspace{5pt} \noindent {\bf (2) $S(k)$ vanishes quadratically $S(k) \sim S_2 k^2$.} In this case, the real-space charge correlation is short-ranged; more precisely, $S_2$ corresponds to the second moment of $\langle \delta n_0 \delta n_r \rangle$. In this case, the integral is not dominated by a small $k$ region; instead, it is dominated by a large $k$ contribution, which generically yields an area-law contribution to the charge fluctuation. Indeed, the IR contribution can be easily shown to give $S_2 l^{d-2}$ contribution, while the UV contribution gives $l^{d-1}$ contribution. To see this, note that $k^2$ factor acts like a discrete differentiation ($\hat k^2=4\sum_j\sin^2(k_j/2)\sim|k|^2$) such that 
\begin{align}
    \frac{S_2}{V} \sum_k k^2 |\tilde w_A(k)|^2 \approx S_2 \sum_r \sum_j |w_A(x + \hat{e}_j) - w_A(x)|^2 \sim S_2 |\rd A|,
\end{align}
since the difference is non-zero only at the boundary of $A$. Physically, when correlations are short-ranged, only degrees of freedom within a correlation length of the boundary can exchange charge with the complement, so number fluctuations are subextensive and boundary controlled.

\vspace{5pt} \noindent {\bf (3) $S(k)$ vanishes linearly $S(k) \sim S_1 |k|$.} In this case, a nonanalytic structure of $S(k)$ indicates that the real-space charge correlation is power-law decaying. Accordingly, on top of area-law charge variance from UV contribution, there is an additional nontrivial IR contribution that scales as $l^{d-1} \log l$~\cite{Fluctuations2012}.
 
Although this gives only subextensive block-charge variance, it can still be sufficient for gaplessness in a local strongly symmetric Lindbladian. The reason is that the variational locality bound on the relaxation rate of a charge mode is controlled by $k^2/S(k)$; hence $S(k)\sim |k|$ gives $\Gamma(k)\lesssim |k|$, which still vanishes as $k\to0$. Thus finite nonzero $S(0+)$ is a sufficient static route to gaplessness, but not a necessary one.

\end{document}